\documentclass[12pt]{paper}
\usepackage[a4paper,hmargin={2.5cm,2.5cm}, vmargin={3.5cm, 3cm}]{geometry}
\usepackage{amsmath,amsfonts,amssymb,amstext,amsthm,amssymb,bbm,hyperref,xcolor}

\theoremstyle{plain}
\newtheorem{thm}{Theorem}
\newtheorem{defin}{Definition}[section]

\newtheorem{prop}[defin]{Proposition}
\newtheorem{cl}[defin]{Claim}
\newtheorem{rmk}[defin]{Remark}

\newtheorem{cor}[defin]{Corollary}

\newtheorem*{thm*}{Theorem}

\usepackage{eucal}
\def\GL{\operatorname{GL}}
\def\Sp{\operatorname{Sp}}
\def\dist{\operatorname{dist}}
\def\Res{\operatorname{Res}}
\def\diam{\operatorname{diam}}
\def\leq{\leqslant}
\def\geq{\geqslant}
\usepackage[notext, upint]{stix}
\usepackage{step}
\usepackage[T1]{fontenc}

\begin{document}

\title{Anderson localisation for quasi-one-dimensional random operators }
\author{Davide Macera\textsuperscript{1} and Sasha Sodin\textsuperscript{2}}
\maketitle

\begin{abstract} In 1990, Klein, Lacroix, and Speis proved (spectral) Anderson localisation for the Anderson model on the strip of width $W \geqslant 1$, allowing for singular distribution of the potential. Their proof employs multi-scale analysis, in addition to arguments from the theory of random matrix products (the case of regular distributions was handled earlier in the works of Goldsheid and Lacroix by other means). We give a proof of their result avoiding multi-scale analysis, and also extend it to the general quasi-one-dimensional  model, allowing,  in particular,  random hopping. Furthermore, we prove a sharp bound on the eigenfunction correlator of the model, which implies exponential dynamical localisation and exponential decay of the Fermi projection. 

The method is also applicable to operatos on the half-line with arbitrary (deterministic, self-adjoint) boundary condition.

Our work generalises and complements the single-scale proofs of localisation in pure one dimension ($W=1$), recently  found by Bucaj--Damanik--Fillman--Gerbuz--VandenBoom--Wang--Zhang,  Jito\-mir\-skaya--Zhu, Go\-ro\-detski--Kleptsyn, and Rangamani. 

\end{abstract}
\footnotetext[1]{Department of Mathematics and Physics, Roma Tre University, Largo San Murialdo 1, 00146 Roma, Italy.  Email: davide.macera@uniroma3.it}
\footnotetext[2]{School of Mathematical Sciences, Queen Mary University of London, London E1 4NS, United Kingdom. Email: a.sodin@qmul.ac.uk. Supported in part by the European Research Council starting grant 639305 (SPECTRUM), a Royal Society Wolfson Research Merit Award (WM170012), and a Philip Leverhulme Prize of the Leverhulme Trust (PLP-2020-064).}

\section{Introduction}

\subsection{The operator, transfer matrices and Lyapunov exponents}
Let $W \geqslant 1$. Let $\{L_x\}_{x \in \mathbb Z}$ be a sequence of   identically distributed $W \times W$ random matrices in $\GL(W, \mathbb R)$, and let $\{ V_x \}_{x \in \mathbb Z}$ be a sequence of  identically distributed $W \times W$ real symmetric matrices, so that $\{L_x\}_{x \in \mathbb Z}, \{V_x\}_{x \in \mathbb Z}$ are jointly independent. Denote by $\mathcal L$ the support of $L_0$ and by $\mathcal V$ -- the support of $V_0$. Throughout this paper we assume that 
\begin{enumerate}
\item[(A)] there exists $\eta > 0$ such that 
\[ \mathbb E (\|V_0 \|^\eta + \| L_0\|^\eta +  \|L_0^{-1}\|^\eta) < \infty~; \]
\item[(B)] the Zariski closure of the group generated by $\mathcal L \, \mathcal L^{-1}$ in $\GL(W, \mathbb R)$ intersects $\mathcal L$ (this holds for example when $\mathbbm 1 \in \mathcal L$);
\item[(C)] $\mathcal V$ is irreducible (i.e.\ has no   common invariant subspaces except for $\{0\}$ and $\mathbb R^W$), and $\mathcal V - \mathcal V$ contains a matrix of rank one. 
\end{enumerate}
We are concerned with the spectral properties of the random operator $H$ acting on (a dense subspace of) $\ell_2 (\mathbb Z \to \mathbb C^W)$ via \begin{equation}\label{eq:defop}
(H \psi)(x) = L_x \psi(x+1) + V_x \psi(x) + L_{x-1}^\intercal \psi(x-1)~, \quad x \in \mathbb Z~.
\end{equation}
This model, often referred to as a quasi-one-dimensional random operator, is the   general Hamiltonian describing a quantum particle with $W$ internal degrees of freedom in random potential and   with nearest-neighbour random hopping. The special case $L_x \equiv \mathbbm 1$ is known as the block Anderson model; it is in turn a generalisation of  the Anderson model on   the strip $\mathbb Z \times \{1, \cdots, W\}$, and, more generally, on $\mathbb Z \times \Gamma$, where $\Gamma$ is any connected finite graph (the assumption that $\Gamma$ is connected ensures that $\mathcal V$ is irreducible). Another known special case of (\ref{eq:defop}) is the Wegner orbital model.  

Fix $E \in \mathbb R$. If $\psi: \mathbb Z \to \mathbb C^W$ is a formal solution of the equation 
\[ L_x \psi(x+1) + V_x \psi(x) + L_{x-1}^\intercal \psi(x-1) = E \psi(x)~, \quad x \geqslant 1~, \] 
then 
\begin{equation}\label{eq:Tx}
\binom{\psi(x+1)}{\psi(x)} = T_x \binom{\psi(x)}{\psi(x-1)}~, \end{equation}
where the one-step  transfer matrix $T_x \in \GL(2W, \mathbb R)$ is given by
\begin{equation}\label{eq:Tx'}
T_x = \left( \begin{array}{cc} L_x^{-1}(E \mathbbm{1} - V_x) & - L_x^{-1} L_{x-1}^\intercal \\ \mathbbm{1} & 0 \end{array}\right)~.
\end{equation}
The multi-step transfer matrices $\Phi_{x, y} \in \GL(2W, \mathbb R)$, $x, y \in \mathbb Z$, are defined by
\begin{equation}\label{eq:Phix}
\Phi_{x,y} = \begin{cases}
T_{x-1} \cdots T_y~, & x > y \\
\mathbbm{1}~, & x = y \\
T_{x}^{-1} \cdots T_{y-1}^{-1}~, & x < y~,
\end{cases}
\end{equation}
so that 
\begin{equation}\label{eq:Phix'}
\Phi_{x,y} \binom{\psi(y)}{\psi(y-1)} = \binom{\psi(x)}{\psi(x-1)}~.
\end{equation}
In particular, $T_x = \Phi_{x+1,x}$. We abbreviate $\Phi_{N} = \Phi_{N,0}$. The Lyapunov exponents $\gamma_j(E)$, $1 \leqslant j \leqslant 2W$, are defined as
\[ \gamma_j(E) = \lim_{N \to \infty} \frac1N \mathbb E \log s_j(\Phi_N(E))~, \]
where $s_j$ stands for the $j$-th singular value. It is known \cite{FK} that (for fixed $E$) this limit in expectation is also an almost sure limit. The cocycle $\{\Phi_{x,y}\}$ is conjugate to a symplectic one (see Section~\ref{s:sympl}), and hence 
\[ \gamma_j (E) = - \gamma_{2W+1-j} (E)~, \quad j =1 ,\cdots, W~.\]
Further, as we shall see in Section~\ref{s:simpl}, using the work of Goldsheid \cite{G95} to verify the conditions of the Goldsheid--Margulis theorem \cite{GM}  on the simplicity of the Lyapunov spectrum, that
\[ \gamma_1(E)  > \gamma_2(E) > \cdots > \gamma_W(E) > 0~.\]
We also mention that the Lyapunov exponents $\gamma_j(E)$ are continuous functions of $E$. This was proved by Furstenberg and Kifer in \cite{FKif};  it can also be deduced from the large deviation estimate (\ref{eq:LDP-norm}) -- see Duarte and Klein \cite{DK}.

\subsection{The main results}

\begin{thm}\label{thm:1}
Assume (A)--(C). Then the spectrum of $H$ is almost surely pure point. Moreover, if 
\[ \mathcal E[H] = \left\{ (E, \psi) \in \mathbb R \times \ell_2(\mathbb Z \to \mathbb C^W) \, : \, \| \psi\| = 1~, \, H
\psi = E\psi \right\} \]
is the collection of eigenpairs of $H$, then  
\begin{equation}\label{eq:thm1} \mathbb P \left\{ \forall (E, \psi) \in \mathcal E[H] \,\,\, \limsup_{x \to \pm\infty} \frac{1}{|x|} \log \|\psi(x)\| \leqslant - \gamma_W(E)\right\}   =1~, \end{equation}
i.e.\ each eigenfunction decays exponentially, with the rate lower-bounded by the slowest Lyapunov exponent.
\end{thm}
\begin{rmk} It is believed that the lower bound is sharp, i.e.\ the rate of decay can not be faster than the slowest Lyapunov exponent:
\begin{equation}\label{eq:lowerbd} \mathbb P \left\{ \forall (E, \psi) \in \mathcal E[H] \,\,\, \liminf_{x \to \pm\infty} \frac{1}{|x|} \log \|\psi(x)\| \geqslant - \gamma_W(E)\right\}   =1.\end{equation}
We refer to \cite{GS} for a discussion and partial results in this direction. For $W=1$, (\ref{eq:lowerbd}) was proved by Craig and Simon in \cite{CS}.
\end{rmk}

The property of having pure point spectrum with exponentially decaying eigenfunctions is a manifestation of Anderson localisation of the random operator $H$.  The mathematical work on Anderson localisation in one dimension was initiated by Goldsheid, Molchanov and Pastur \cite{GMP}, who considered the case $W=1$, $L_x \equiv 1$ and established the pure point nature of the spectrum under the assumption that the distribution of $V_x$ is regular enough (absolutely continuous with bounded density). A different proof of the result of \cite{GMP} was found by  Kunz and Souillard \cite{KS}. Under the same assumptions, the exponential decay of the eigenfunctions was established by Molchanov \cite{Molch}. The case of singular distributions  was treated by 
 Carmona, Klein, and Martinelli \cite{CKM}. 

The case $W >1$ was first considered by Goldsheid \cite{G80}, who established the pure point nature of the spectrum for the case of the Schr\"odinger operator on the strip, i.e.\ when $L_x \equiv \mathbbm 1$, $V_x$ is tridiagonal with the off-diagonal entries equal to $1$ and the diagonal ones independent and identically distributed, under the assumption that the distribution of the diagonal etries of $V_x$ is regular. In the same setting, Lacroix \cite{Lacr,Lacr2,Lacr3} proved that the eigenfunctions decay exponentially. The case of the Anderson model on a strip with general (possibly, singular) distributions was settled by Klein--Lacroix--Speis \cite{KLS}, who  established localisation in the strong form (\ref{eq:thm1}).

Unlike the earlier, more direct arguments treating regular distributions, the works \cite{CKM, KLS} allowing singular distributions involve a multi-scale argument (as developed in the work of Fr\"ohlich--Spencer \cite{FS} on localisation in higher dimension); the theory of random matrix products is used to verify the initial hypothesis of multi-scale analysis. Recently, proofs of  the result of \cite{CKM} avoiding multi-scale analysis were found by Bucaj et al.\ \cite{BDFGVWZ}, Jitomirskaya and Zhu \cite{JZ}, and Gorodetski and Kleptsyn \cite{GorKl}; the general one-dimensional case (allowing for random hopping) was settled by Rangamani \cite{R}.  Our Theorem~\ref{thm:1} can be seen as a generalisation of 
these works, and especially of \cite{JZ,R}, to which our arguments are closest in spirit: we give a relatively short and single-scale proof of localisation which applies to arbitrary $W \geqslant 1$, and allows for rather general distributions of $V_0$ and $L_0$ (under no regularity assumptions on the distribution of the potential). In particular, we recover and generalise the result of \cite{KLS}.

In fact, we prove a stronger result pertaining to the eigenfunction correlators, introduced by Aizenman \cite{A} (see further the monograph of Aizenman--Warzel \cite{AW}). If $\Lambda \subset \mathbb Z$ is a finite set, denote by $H_\Lambda$ the restriction of $H$ to $\ell_2(\Lambda \to \mathbb C^W)$, i.e. 
\[ H_\Lambda = P_\Lambda H P_\Lambda^*~, \]
where $P_\Lambda: \ell_2(\mathbb Z \to \mathbb C^W) \to \ell_2(\Lambda \to \mathbb C^W)$ is the coordinate projection. If $I \subset \mathbb R$ is a compact interval, denote
\[ Q_I^\Lambda(x, y) = \sup \left\{ \|f(H_{\Lambda})_{x,y}\| \, : \, \operatorname{supp} f \subset I~, \, |f| \leqslant 1 \right\}~, \quad Q_I(x, y) = \sup_{a \leqslant x, y \leqslant b}  Q_I^{[a,b]}(x, y)~. \]
Here $\|f(H_{\Lambda})_{x,y}\|$ is the operator norm of the $(x,y)$ block of $f(H_{\Lambda})$, and the functions $f$ in the supremum are assumed to be, say, Borel measurable.

\begin{thm}\label{thm:2} Assume (A)--(C). For any compact interval $I \subset \mathbb R$, 
\begin{equation}\label{eq:thm2}\mathbb P\left\{  \limsup_{x \to \pm \infty} \frac{1}{|x|} \log Q_I(x, y) \leqslant - \inf_{E \in I} \gamma_W(E) \right\} = 1~.\end{equation}
\end{thm}
It is known (see \cite{AW}) that Theorem~\ref{thm:2} implies Theorem~\ref{thm:1}. By plugging in various choices of $f$, it also implies (almost sure) dynamical localisation with the sharp rate of exponential decay, the exponential decay of the Fermi projection, et cet. (see e.g. \cite{AG} and \cite{AW}). We chose to state Theorem~\ref{thm:1} as a separate result rather than a corollary of Theorem~\ref{thm:2} since its direct proof is somewhat shorter than that of the latter. 

We refer to Bucaj et al.\ \cite{BDFGVWZ}, Jitomirskaya--Zhu \cite{JZ}, and Ge-Zhao \cite{GZ} for earlier results on dynamical localisation for $W = 1$.

\subsection{Main ingredients of the proof}\label{s:ingr}

Similarly to many of the previous works, including \cite{CKM, KLS} and also the recent works \cite{BDFGVWZ,JZ,GorKl}, the two main ingredients of the proof of localisation are a large deviation estimate and a Wegner-type estimate. We state these in the  generality required here. Let $I \subset \mathbb R$ be a compact interval, and let $F \subset \mathbb R^{2W}$ be 
a Lagrangian subspace (see Section~\ref{s:tm}). Denote by $\pi_F: \mathbb R^{2W} \to F$ the orthogonal projection onto $F$.
\begin{prop}\label{prop:LDP} Assume (A)--(C). For any $\epsilon > 0$ there exist $C, c> 0$ such that for any $E \in I$ and any Lagrangian subspace $F \subset \mathbb R^{2W}$
\begin{equation}\label{eq:LDP}
\mathbb P \left\{ \left| \frac1N \log s_W(\Phi_N(E) \pi_F^*) - \gamma_W(E) \right| \geqslant \epsilon \right\} \leqslant C e^{-cN}~.
\end{equation}
\end{prop}
The proof is essentially given in \cite{KLS}; we outline the necessary reductions in Section~\ref{s:sympl}. The second proposition could be also proved along the lines of the special case considered in \cite{KLS}; we present an alternative (arguably, simpler) argument in Section~\ref{s:Weg}.

For an operator $H$ and $E$ in the resolvent set of $H$, we denote by $G_E[H] = (H-E)^{-1}$ the resolvent of $H$ and by $G_E[H](\cdot, \cdot)$ its matrix elements. If $E$ lies in the spectrum of $H$, we set $G_E[H](\cdot, \cdot) \equiv  \infty$.
\begin{prop}\label{prop:Weg} Assume (A)--(C). For any $\epsilon > 0$ there exist $C, c> 0$ such that for any $E \in I$ and $N \geqslant 1$  
\begin{equation} \label{eq:ourweg}\begin{split}
\mathbb P \left\{  \|G_E[H_{[-N,N]}](i, i)   \|\leqslant e^{-\epsilon N} \right\} \leqslant C e^{-cN}  \quad &(i \in [-N,N])\\
\mathbb P \left\{  \|G_E[H_{[-N,N]}](i, i\pm 1)    \|\leqslant e^{-\epsilon N} \right\} \leqslant C e^{-cN}  \quad &(i , i \pm 1\in [-N,N])
\end{split}
\end{equation}
\end{prop}
\begin{rmk} The arguments which we present can be applied to deduce the following strengthening of (\ref{eq:ourweg}):
\[ \mathbb P \left\{ \dist(E, \sigma(H_{[-N,N]})) \leqslant e^{-\epsilon N} \right\} \leqslant C e^{-cN}~. \]
We content ourselves with (\ref{eq:ourweg}) which suffices for the proof of the main theorems.
\end{rmk}

Klein, Lacroix and Speis \cite{KLS} use (special cases of) Propositions~\ref{prop:LDP} and \ref{prop:Weg} to verify the assumptions required for multi-scale analysis. We deduce Theorems~\ref{thm:1} and \ref{thm:2} directly from these propositions. In this aspect, our general strategy is similar to the cited works \cite{BDFGVWZ,JZ,GorKl}. However, several of the arguments employed in these works rely on the special features of the model for $W=1$; therefore our implementation of the strategy differs in several crucial aspects.

\paragraph{Acknowledgement} We are grateful to Ilya Goldsheid for helpful discussions,  and to Alexander Elgart for spotting a number of lapses in a preliminary version of this paper.

\section{Proof of the main theorems}

\subsection{Resonant sites; the main technical proposition}

Let $\tau > 0$ be a (small) number. We say that $x\in \mathbb Z$ is $(\tau, E, N)$-non-resonant ($x \notin\Res(\tau, E,N)$) if 
\begin{equation}\label{eq:nonres}\begin{cases}
\|L_x\| \leqslant e^{\tau N}~, \\
\|G_E[H_{[x-N, x+N]}](x, x\pm N)\| \leqslant e^{-(\gamma_W(E) - \tau)N}~,
\end{cases}\end{equation}
and $(\tau, E,N)$-resonant ($x \in\Res(\tau, E,N)$) otherwise.  The following proposition is the key step towards the proof of Theorems~\ref{thm:1} and \ref{thm:2}.

\begin{prop}\label{prop:main} Assume (A)--(C). Let $I \subset \mathbb R$ be a compact interval, and let $\tau > 0$. There exist $C, c> 0$ such that for any $N \geqslant 1$
\[ \mathbb P \left\{ \max_{E \in I} \diam (\Res(\tau, E, N) \cap [-N^2, N^2]) > 2N \right\} \leqslant C e^{-cN}~. \]
\end{prop}

The remainder of this section is organised as follows. In Section~\ref{s:redtm}, we express the Green function in terms of the transfer matrices. Using this expression and Propositions~\ref{prop:LDP} and \ref{prop:Weg}, we  show that the probability that $x \in \Res(\tau, E, N)$ (for a fixed $E \in \mathbb R$) is exponentially small. In Section~\ref{s:pfmainprop}, we rely on this estimate to prove Proposition~\ref{prop:main}. Then we use this proposition to prove Theorem~\ref{thm:1} (Section~\ref{s:pf1}) and Theorem~\ref{thm:2} (Section~\ref{s:pf2}).

\subsection{Reduction to transfer matrices}\label{s:redtm}

Fix $N \geqslant 1$. Consider the $W \times W$ matrices
\begin{equation}\label{eq:psi}
\begin{split}
&\Psi_i^+ = (\mathbbm1 \,\, 0) \, \Phi_{i, N+1} \binom{0}{\mathbbm 1} = (0 \,\, \mathbbm 1) \Phi_{i+1,N+1} \binom{0}{\mathbbm 1}~, \\ 
&\Psi_i^- = (\mathbbm1 \,\, 0) \, \Phi_{i, -N} \binom{\mathbbm 1}{0} = (0 \,\, \mathbbm 1) \Phi_{i+1,-N} \binom{\mathbbm 1}0~. 
\end{split}\end{equation}
The Green function of $H_{[-N,N]}$ can be expressed in terms of these matrices using the following claim, which holds deterministically for any $H$ of the form (\ref{eq:defop}). A similar expression has been employed already in \cite{G80}.
\begin{cl}\label{cl:formG} If $E \notin \sigma(H_{[-N,N]})$, then:
\begin{enumerate}
\item 
\[\binom{\Psi_{\pm1}^\pm}{\Psi_{0}^\pm} G_E[H_{[-N,N]}](0, \pm N)
= \binom{G_E[H_{[-N,N]}](0, \pm 1)}{G_E[H_{[-N,N]}](0,0)} ;\]

\item for any $i,j \in [-N, N]$,
\[ G_E[H_{[-N, N]}](i, i) = \begin{cases}
\Psi_j^{-} (\Psi_i^-)^{-1} \left( \Psi_{i+1}^+ (\Psi_i^+)^{-1} - \Psi_{i+1}^- (\Psi_i^-)^{-1}\right)^{-1} L_i^{-1}~, &i \geq j\\
\Psi_j^{+} (\Psi_i^+)^{-1}  \left( \Psi_{i+1}^+ (\Psi_i^+)^{-1} - \Psi_{i+1}^- (\Psi_i^-)^{-1}\right)^{-1} L_i^{-1}~, &i \leq j~.
\end{cases}\]
\end{enumerate}
\end{cl}

\begin{proof} Abbreviate $G_E = G_E[H_{[-N,N]}]$, and  set $G_E(i, j) = 0$ for $j \notin [-N, N]$. The matrices $G_E(i, j)$, $-N \leqslant j \leqslant N$, are uniquely determined by the system of equations
\begin{equation}\label{eq:condG}
L_j G_E(i, j+1) + (V_j - E \mathbbm 1) G_E(i, j) + L_{j-1}^\intercal G_{E}(i,j-1) = \delta_{j,i} {\mathbbm1}~, \quad -N \leqslant j \leqslant N~.
\end{equation}
We look for a solution of the form
\begin{equation}\label{eq:GviaPsi} G_E(i, j) = \begin{cases}
\Psi_j^- \alpha_i^-~, j\leqslant i \\
\Psi_j^+ \alpha_i^+~, j\geqslant i~,
\end{cases}\end{equation}
where 
\begin{eqnarray}
\Psi_i^- \alpha_i^- - \Psi_i^+ \alpha_i^+  &=& 0 \label{eq:al1} \\
\Psi_{i+1}^- \alpha_i^- - \Psi_{i+1}^+ \alpha_i^+ &=& - L_i^{-1}~. \label{eq:al2}
\end{eqnarray}
The first equation ensures that (\ref{eq:GviaPsi}) defines $G_E(i,i)$ consistently, while the second one guarantees that (\ref{eq:condG}) holds for $j=i$. For the other values of $j$, (\ref{eq:condG}) follows from the construction of the matrices $\Psi^{\pm}_j$. 

The solution to (\ref{eq:al1})--(\ref{eq:al2}) is explicitly found by elimination:
\[ \alpha_i^- = (\Psi_i^-)^{-1} \Psi_i^+ \alpha_i^+~, \quad 
\alpha_i^+ = - ( \Psi_{i+1}^- (\Psi_i^-)^{-1} \Psi_i^+ - \Psi_{i+1}^+)^{-1} L_i^{-1}~.
\]
This implies the second part of the claim. For the first part, note that for $j \geqslant i$
\[ G_E(0,j) = \Psi_j^+ \alpha_0^+ = \Psi_j^+ (\Psi_0^+)^{-1} G_E(0,0) =  \Psi_j^+ (\Psi_1^+)^{-1} G_E(0,1)~.\]
Observing that $\Psi_N^+ = \mathbbm1$, we conclude that 
\[  G_E(0,N) =  (\Psi_0^+)^{-1} G_E(0,0) =    (\Psi_1^+)^{-1} G_E(0,1)~,\]
as claimed.  Similarly,
\[  G_E(0,-N) =  (\Psi_0^-)^{-1} G_E(0,0) =    (\Psi_{-1}^-)^{-1} G_E(0,-1)~. \qedhere\]
\end{proof}

\subsection{Proof of Proposition~\ref{prop:main}}\label{s:pfmainprop}

Fix a small $\tau > 0$. Without loss of generality $I$ is short enough to ensure that 
\[ \max_{E \in I} \gamma_W(E) - \min_{E \in I} \gamma_W(E) \leqslant \frac\tau 2\]
(this property is valid for short intervals due to the continuity of $\gamma_W$; the statement for larger intervals $I$ follows by compactness). Fix such $I$ (which will be suppressed from the notation), and let
\[ \gamma = \frac12 (\max_{E \in I} \gamma_W(E) + \min_{E \in I} \gamma_W(E))~, \quad \text{so that} 
\quad \sup_{E \in I} | \gamma_W(E) - \gamma| \leqslant \frac\tau4~.\]
For $x \in \mathbb Z$, let  
\[ \Res^*(\tau, x, N) = \left\{ E \in I \, : \, \max_\pm \|G_E[H_{[x-N,x+N]}](x,x\pm N) \|_{1,\infty} \geqslant e^{-(\gamma(E) - \frac\tau2)N} \right\}~,\]
where $\|A\|_{1,\infty} = \max_{1 \leqslant \alpha, \beta \leqslant W} |A_{\alpha,\beta}|$. For $N$ large enough ($N \geqslant N_0(\tau)$),
\[ \left( \| L_x \| \leqslant e^{\tau N} \right) \,\, \text{and} \,\, \left( E \notin \Res^*(\tau, x, N) \right) \Longrightarrow x \notin \Res(\tau, E, N)~.\]
By (A) and the Chebyshev inequality
\[ \mathbb P \left\{ \exists x \in [-N^2, N^2] \, : \, \|L_x\| \geqslant e^{\tau N} \right\}
\leqslant (2N^2 + 1) \frac{\mathbb E \|L_0\|^\eta}{e^{\tau \eta N}} \leqslant C_1 e^{-c_1N}~. \]
Hence the proposition boils down to the following statement:
\begin{equation}\label{eq:intersect}
|x-y| >2N \Longrightarrow 
\mathbb P \left\{ \Res^*(\tau, x, N) \cap \Res^*(\tau, y, N) \neq \varnothing \right\} \leqslant C e^{-cN}~.
\end{equation}
The proof of (\ref{eq:intersect}) rests on two claims. The first one is deterministic:

\begin{cl}\label{cl:un} $\Res^*(\tau, x, N)$ is the union of at most $C_W N$ disjoint closed intervals.
\end{cl}
\begin{proof}
By Cramer's rule, for each $\alpha,\beta \in \{1, \cdots, W\}$ and $\pm$ the function
\[ g_{\alpha,\beta}^\pm: E \mapsto (G_E[H_{[x-N,x+N]}](x, x\pm N))_{\alpha,\beta} \]
is the ratio of two polynomials of degree $\leqslant W(2N+1)$. Hence the level set 
\[ \left\{ E \, : \,  |g_{\alpha,\beta}^\pm(E)| = e^{-(\gamma - \frac \tau 2)N} \right\} \]
is of cardinality $\leqslant W(2N+1)$ (note that the $\leqslant W(2N+1)$ discontinuity points of $g_{\alpha,\beta}^\pm$ are poles, hence they can not serve as the endpoints of the superlevel sets of this function). Hence  our set
\[ \left\{ E \, : \,  |g_{\alpha,\beta}^\pm(E)| \geqslant e^{-(\gamma - \frac \tau 2)N} \right\} \]
is the union of at most $\leqslant W(2N+1)/2$ closed intervals, and $\Res^*(\tau, x, N)$ is the union of at most 
\[ 2 \, \frac{W(W+1)}{2} \, \frac{W(2N+1)}{2} \leqslant C_W N\]
closed intervals.
\end{proof}

\begin{cl}\label{cl:prob} Assume (A)--(C). For any compact interval $I \subset \mathbb R$ there exist $C, c> 0$ such that for any $N \geqslant 1$ and any $E \in I$, 
\[ \mathbb P\left\{ E \in \Res^*(\tau, x, N) \right\} \leqslant C e^{-cN}~.\]
\end{cl}
\begin{proof} According to Claim~\ref{cl:formG},
 \[ \| G_E[H_{[-N,N]}](0, \pm N)\| \leqslant \left\{ s_W \binom{\Psi_{\pm1}^\pm}{\Psi_{0}^\pm} \right\}^{-1} \,  \| \binom{G_E[H_{[-N,N]}](0, \pm 1)}{\|G_E[H_{[-N,N]}](0,0)}\|~;\]
hence
\[ \begin{split}
&\mathbb P \left\{  \|G_E[H_{[-N,N]}](0, \pm N) \| \geqslant e^{-(\gamma_W(E) - \frac\tau4)N)} \right\} \\
&\quad\leqslant \mathbb P\left\{ s_W \binom{\Psi_{\pm 1}^\pm}{\Psi_0^\pm} \leqslant e^{(\gamma_W(E) - \frac\tau8)N } \right\} 
+ \mathbb P\left\{\| \binom{G_E[H_{[-N,N]}](0, \pm 1)}{\|G_E[H_{[-N,N]}](0,0)}\| \geqslant e^{\frac{\tau}{8} N} \right\} ~.
\end{split}\]
By Propositions~\ref{prop:LDP} and \ref{prop:Weg}, both terms decay exponentially in $N$,
locally uniformly in $E$. 
\end{proof}

Now we can prove (\ref{eq:intersect}). By Claim~\ref{cl:un} both $\Res^*(\tau, x, N)$ and $\Res^*(\tau, y, N)$ are unions of at most $C_W N$ closed intervals. If these two sets intersect, then either one of the edges of the intervals composing the first one lies in the second one, or vice versa. The operators $H_{[x-N, x+N]}$ and $H_{[y-N, y+N]}$ are independent due to the assumption $|x-y|> 2N$, hence by Claim~\ref{cl:prob}  
\[ \mathbb P \left\{ \Res^*(\tau, x, N) \cap \Res^*(\tau, y, N) \neq \varnothing \right\}  \leqslant 4 C_W N \times C e^{-cN} \leqslant C_1 e^{-c_1 N}~. \]
This concludes the proof of (\ref{eq:intersect}) and of Proposition~\ref{prop:main}.  \qed

\subsection{Spectral localisation: proof of Theorem~\ref{thm:1}}\label{s:pf1}

The proof of localisation is based on Schnol's lemma, which we now recall (see \cite{Han} for a version applicable in the current setting). A function $\psi: \mathbb Z \to \mathbb C^W$ is called a generalised eigenfunction corresponding to a generalised eigenvalue $E \in \mathbb R$ if 
\begin{eqnarray} 
&L_x \psi(x+1) + V_x \psi(x) + L_{x-1}^\intercal \psi(x-1) = E \psi(x)~, \quad x \geqslant 0  \label{eq:efeq}\\
&\limsup_{|x| \to \infty} \frac{1}{|x|} \log \|\psi(x)\| = 0~.\label{eq:grgef}
\end{eqnarray}
Schnol's lemma asserts that any spectral measure of $H$ is supported on the set of generalised eigenvalues. Thus we need to show that (with full probability) any generalised eigenpair $(E, \psi)$ satisfies
\begin{equation}\limsup_{|x| \to \infty} \frac{1}{|x|} \log \|\psi(x)\| \leqslant - \gamma_W(E)~. \end{equation}

Fix a compact interval $I \subset \mathbb R$, and $\tau > 0$. Consider the events 
\[   \mathcal G_M(I, \tau) = \left\{ \forall E \in I \,\,  \forall N \geqslant M  \,\, \diam  ( \Res(\tau, E, N) \cap [-N^2, N^2])\leqslant 2N \right\}~. \]
By Proposition~\ref{prop:main} and the Borel--Cantelli lemma,
\[ \mathbb P \left(\bigcup_{M \geqslant 1}   \mathcal G_M(I, \tau)\right) = 1~. \]
We shall prove that on any $\mathcal G_M(I, \tau)$ every generalised eigenpair $(E, \psi)$ with $E \in I$ satisfies 
\begin{equation}\label{eq:need4tau} \limsup_{|x| \to \infty} \frac{1}{|x|} \log \|\psi(x)\| \leqslant - \gamma_W(E) +3\tau~. \end{equation}
From (\ref{eq:efeq}), we have for any $x$
\[\psi(x) = - G_E[H_{[x-N, x+N]}] (x, x-N) L_{-N-1}^\intercal \psi(x-N-1) - G_E[H_{[x-N, x+N]}] (x, x+N) L_{N} \psi(x+N+1)~.\]
If $x \notin \Res(\tau, E, N)$, this implies
\[\begin{split}
\| \psi(x)\| &\leqslant e^{-(\gamma_W(E) - 2\tau) N} (\| \psi(x-N-1)\| + \|\psi(x+N+1)\|) \\
&\leqslant 2e^{-(\gamma_W(E) - 2\tau) N} \max(\| \psi(x-N-1)\|, \|\psi(x+N+1)\|)~,\end{split}\]
whence $f_\tau(x) \overset{\text{def}}= e^{-\tau |x|} \|\psi(x)\|$ satisfies
\begin{equation}\label{eq:subh}f_\tau(x) 
\leqslant 2 e^{-(\gamma_W(E) - 3\tau) N}\max(f_\tau(x-N-1), f_\tau(x+N+1)))~.\end{equation}
The function $f_\tau$ is bounded due to (\ref{eq:grgef}), hence it achieves a maximum at some $x_\psi \in \mathbb Z$. For 
\[ N > \log 2 / (\gamma_W(E) - 3\tau)~,\]
 (\ref{eq:subh}) can not hold at $x = x_\psi$, thus on $\mathcal G_M(I, \tau)$ x for all 
\[ N \geqslant N_0 \overset{\text{def}}{=} \max(M, \log 2 / (\gamma_W(E) - 3\tau), |x_\psi|) \]
we have:
\[ \Res(\tau, E, N) \cap [-N^2, N^2] \subset [x_\psi - 2N, x_\psi + 2N] \subset [-3N, 3N]~.\]
Thus (\ref{eq:subh}) holds whenever $x, N$ are such that $3N < |x| \leqslant N^2$ and $N \geqslant N_0$.

For each $x \in \mathbb Z$, let $N(x)$ be such that $N^2/10 \leqslant |x| \leqslant N^2 / 5$. If $|x|$ is large enough, $N(x) \geqslant N_0$. Applying (\ref{eq:subh}) $\lfloor |x|/(N+1) \rfloor - 4$ times, we obtain
\[ f_\tau(x) \leqslant  (2e^{-(\gamma_W(E)-3\tau)N})^{\lfloor x/(N+1)\rfloor - 4} \times \max f_\tau 
\leqslant e^{-(\gamma_W(E)-3\tau)|x| + C(\sqrt{|x|}+1)} \times \max f_\tau~, \]
which implies (\ref{eq:need4tau}).
\qed

\subsection{Eigenfunction correlator: proof of Theorem~\ref{thm:2}}\label{s:pf2}

Fix a compact interval $I \subset \mathbb R$, and let $\gamma = \min_{E \in I} \gamma_W(E)$. 
The proof of (\ref{eq:thm2}) relies on the following fact from \cite[Lemma 4.1]{ESS}, based on an idea from \cite{AW}:
\begin{equation}\label{eq:Qint}  Q_I^\Lambda (x, y) \leqslant \lim_{\epsilon \to + 0} \frac{\epsilon}{2} \int_I \! \| G_E[H_\Lambda](x, y)\|^{1-\epsilon} dE \leqslant W~.\end{equation}
Our goal is to bound on this quantity uniformly in the interval $\Lambda \supset \{x,y\}$. Without loss of generality we can assume that $x = 0$. Choose $N$ such that $N^2/10 \leqslant |y| \leqslant N^2/5$. By Proposition~\ref{prop:main}, for any $\tau \in (0, \gamma)$
\[ \mathbb P \left\{ \forall E \in I \, \diam(\Res(\tau, E, N) \cap [-N^2, N^2]) \leqslant 2N\right\} \geqslant 1 - Ce^{-cN}~.  \]
We show that on the  event 
\begin{equation}\label{eq:event} \left\{ \forall E \in I \, \diam(\Res(\tau, E, N) \cap [-N^2, N^2]) \leqslant 2N\right\} \end{equation}
we have
\begin{equation}\label{eq:needQ}
Q_I^\Lambda(0,y) \leqslant e^{-(\gamma - 2\tau)|y|}~, \quad |y| > C_0(\gamma-\tau)~.
\end{equation}
Expand the Green function $G_E[H_\Lambda](0, y)$ as follows. First,  iterate the resolvent identity
\[\begin{split}
G_E[H_\Lambda](x, y) &= G_E[H_{[x - N, x + N]}](x, x-N) L_{x-N-1}^\intercal G_E[H_\Lambda](x-N-1, y) \\
&+G_E[H_{[x - N, x + N]}](x, x + N) L_{x+N}\,\,\,\,\,\,\,G_E[H_\Lambda](x+N+1, y) \end{split} \]
starting from $x = 0$ at most $|y|/N$ times, or until the first argument of $G_E[H_\Lambda]$ reaches the set $\Res(\tau, E, N)$. Then apply the identity 
\[\begin{split}
G_E[H_\Lambda](x, u) &= G_E[H_\Lambda](x, u - N -1) L_{u-N-1}G_E[H_{[u-N,u+N]}](u-N,  u) \\
&+ G_E[H_\Lambda](x, u +N +1) L_{u+N}^\intercal \,\,\,\,\,\,\,G_E[H_{[u-N, u+N]}](u+N,  u)  \end{split} \]
starting from $u = y$ at most $|y|/N$ times, or until the second argument of $G_E[H_\Lambda]$ reaches the set $\Res(\tau, E, N)$. 
The resulting expansion has $\leqslant 2^{2|y|/N}$ addends, each of which has the form
\begin{equation}\label{eq:gfexp}\begin{split}
& G_E[H_{[x_0 - N, x_0 + N]}](x_0, x_1) \cdots G_E[H_{[x_{k-1} - N, x_{k-1} + N]}](x_{k-1}, x_k) \\
&\qquad G_E[H_\Lambda](x_k, y_\ell) \\
&\qquad G_E[H_{[y_{\ell-1} - N, y_{\ell-1} + N]}](y_\ell, y_{\ell-1}) \cdots G_E[H_{[y_0 - N, y_0 + N]}](y_1, y_0)~, 
\end{split}
\end{equation}
where $x_0 = 0$, $x_{j+1} = x_j \pm N$, $y_0 = y$, $y_{j+1} = y_j\pm N$, and (by the construction of the event (\ref{eq:event})) $k + \ell \geqslant |y|/N-4$. All the terms in the first and third line of (\ref{eq:gfexp}) are bounded in norm by $e^{-(\gamma-\tau)N}$, hence 
\[ \| G_E[H_\Lambda](0, y) \| \leqslant 64 \left(4 e^{-(\gamma-\tau)N}\right)^{|y|/N-4} \sum_{u, v \leqslant 2|y|} \|G_E[H_\Lambda](u, v)\|~.\]
Now we raise this estimate to the power $1 -\epsilon$ and integrate over $E \in I$:
\[ \frac{\epsilon}{2} \int_I  \| G_E[H_\Lambda](0, y) \|^{1-\epsilon} dE \leqslant 64^{1-\epsilon} \left(4 e^{-(\gamma-\tau)N}\right)^{(1-\epsilon)(|y|/N-4)} \sum_{u, v \leqslant 2|y|} \frac\epsilon2 \int_I \|G_E[H_\Lambda](u, v)\|^{1-\epsilon} dE~.\]
It remains to let $\epsilon \to + 0$ while making use of the two inequalities in (\ref{eq:Qint}).
 \qed

\section{Properties of transfer matrices}\label{s:tm}

\subsection{Preliminaries}\label{s:sympl}

Denote 
\[ J = \left( \begin{array}{cc} 0 & - \mathbbm 1 \\ \mathbbm 1 & 0 \end{array}\right) \in \GL(2W, \mathbb R)~. \]
A matrix $Q \in \GL(2W, \mathbb R)$ is called symplectic, $Q \in \Sp(2W, \mathbb R)$, if $Q^\intercal J Q = J$. 

The matrices $T_x$ are, generally speaking, not  symplectic. However, the cocycle $\{\Phi_{x,y}\}_{x,y,\in\mathbb Z}$ is  conjugate to a  symplectic one. Indeed, observe that 

\begin{cl} If $L \in \GL(W, \mathbb R)$ and $Z$ is $W \times W$ real symmetric, then $Q(L, Z) = \left( \begin{array}{cc} L^{-1} Z  & -L^{-1} \\ L^\intercal & 0 \end{array} \right)$ is  symplectic.
\end{cl}

Denote $D_x =\left( \begin{array}{cc} \mathbbm 1  & 0 \\ 0 & L_x^\intercal \end{array} \right)$, then
\[ \widetilde T_x(E) \overset{\text{def}}= D_x T_x(E) D_{x-1}^{-1} = Q(L_x, E\mathbbm 1-V_x) \in \Sp(2W, \mathbb R)~. \]
Thus also 
\[ \widetilde \Phi_{x, y}(E) = D_{x-1} \Phi_{x, y}(E) D_{y-1}^{-1} = 
  \begin{cases}
\widetilde T_{x-1} (E)\cdots \widetilde T_y(E)~, & x > y \\
\mathbbm{1}~, & x = y \\
\widetilde T_{x}^{-1} (E)\cdots T_{y-1}^{-1}(E)~, & x < y 
\end{cases} \,\,\,\,\,\, \in \Sp(2W, \mathbb R)~. \]

\subsection{Simplicity of the Lyapunov spectrum and large deviations}\label{s:simpl}

Goldsheid and Margulis showed \cite{GM} that if $g_j$ are independent, identically distributed random matrices in $\Sp(2W, \mathbb R)$, and the group generated by the support of $g_1$ is Zariski dense in $\Sp(2W, \mathbb R)$, then the Lyapunov spectrum of a random matrix product $\{ g_N \cdots g_1 \}$ is simple, i.e.\ 
\[ \gamma_1 > \cdots > \gamma_W> 0~.\]
Goldsheid showed \cite{G95} that if $\mathcal V$ is irreducible and $\mathcal V - \mathcal V$ contains a rank-one matrix, then for any $E \in \mathbb R$ the group generated by $Q(\mathbbm 1, E \mathbbm 1 - V)$, $V \in \mathcal V$, is Zariski dense in $\Sp(2W, \mathbb R)$. 
\begin{cor} Assume (A)--(C). Then for any $E \in \mathbb R$
\[ \gamma_1(E) > \cdots > \gamma_W(E) > 0~.\]
\end{cor}
\begin{proof}
Observe that 
\[ Q(L, E\mathbbm 1-V) = \left( \begin{array}{cc} L^{-1} & 0 \\ 0 & L^\intercal \end{array} \right) Q(\mathbbm 1,  E\mathbbm 1-V)~,\]
whence
\[ Q(\widehat L, E\mathbbm 1-V)^{-1} Q(L, E\mathbbm 1-V) = \left( \begin{array}{cc} \widehat L L ^{-1} & 0 \\ 0 & \widehat L^{-\intercal}L^\intercal \end{array} \right)~.\]
If the Zariski closure of the group generated by $\mathcal L \mathcal L^{-1}$ intersects $\mathcal L$, then  the Zariski closure of the group generated by $\{ Q(L, E \mathbbm 1 - V)\}_{L \in \mathcal L, V \in \mathcal V}$ contains that of the group generated by $\{ Q(\mathbbm 1, E \mathbbm 1 - V)\}_{V \in \mathcal V}$. 
\end{proof}
Having the corollary at hand, we deduce from \cite[Proposition 2.7]{KLS} applied to the matrices $\widetilde \Phi_N(E)$:
\begin{prop}\label{prop:LDP+} Assume (A)--(C). For any $\epsilon > 0$ there exist $C, c> 0$ such that for any $E \in I$ and $1 \leqslant j \leqslant W$
\begin{equation}\label{eq:LDP-norm}
\mathbb P \left\{ \left| \frac1N \log s_j(\widetilde \Phi_N(E)) - \gamma_j(E) \right| \geqslant \epsilon \right\} \leqslant C e^{-cN}~.
\end{equation}
Further, for any Lagrangian subspace $F \subset \mathbb R^{2W}$
\begin{equation}\label{eq:LDP-vec}
\mathbb P \left\{ \left| \frac1N \log s_j(\widetilde \Phi_N(E) \pi_F^*) - \gamma_j(E) \right| \geqslant \epsilon \right\} \leqslant C e^{-cN}~.
\end{equation}
\end{prop}
\begin{proof} The estimate (\ref{eq:LDP-vec}) is a restatement of \cite[Proposition 2.7]{KLS}, whereas (\ref{eq:LDP-norm}) follows from (\ref{eq:LDP-vec}) applied to a $\delta$-net on the manifold of Lagrangian subspaces of $\mathbb R^{2W}$ (the Lagrangian Grassmannian). We note that (\ref{eq:LDP-norm}) is also proved directly in \cite{DK}.
\end{proof}
Note that Proposition~\ref{prop:LDP} follows from (\ref{eq:LDP-vec}). 

\medskip 
Now fix $\epsilon$ and a Lagrangian subspace $F$, and let 
\begin{equation}\label{eq:defomega} \Omega_\epsilon^F[\widetilde \Phi_N]  = \left\{ \max_{j=1}^W   \left[ | \frac1N \log s_j(\widetilde \Phi_N) - \gamma_j  |  +  | \frac1N \log s_j(\widetilde \Phi_N  \pi_F^*) - \gamma_j | \right] \leqslant \frac{\epsilon}{100 W} \right\}~.\end{equation}
According to Proposition~\ref{prop:LDP+}, 
\[ \mathbb P(\Omega_\epsilon^{F}[\widetilde \Phi_N(E)]) \geqslant 1 - C(\epsilon, E) e^{-c(\epsilon, E) N}~, \]
where the constants are locally uniform in $E$.
Let 
\[ \widetilde \Phi_N(E) = U_N(E) \Sigma_N(E) V_N(E)^\intercal \]
be the singular value decomposition of $\widetilde \Phi_N(E)$. Assume that the singular values on the diagonal of $\Sigma_N(E)$ are arranged in non-increasing order; the choice of the additional degrees of freedom is not essential for the current discussion. Denote
\begin{equation}\label{eq:Fpm} F_+ = \left\{ \binom{x}{0} \, :\, x \in \mathbb R^{W} \right\} \subset \mathbb R^{2W}~, \quad 
F_- = \left\{ \binom0y \, :\, y \in \mathbb R^{W} \right\} \subset \mathbb R^{2W}~. \end{equation}
\begin{cl}\label{cl:unit} Let $F \subset \mathbb R^{2W}$ be a Lagrangian subspace. For $N$ large enough (depending on $\epsilon$), one has (deterministically) on the event $\Omega_\epsilon^F[\widetilde \Phi_N(E)]$ defined in (\ref{eq:defomega})
\[ s_W(\pi_{F_+ }V_N(E)^\intercal \pi_F^*) \geqslant e^{-\frac{\epsilon}{25} N}~.  \]
\end{cl}
\begin{rmk} For future reference, we also record the dual version of the claim: on $\Omega_\epsilon^F[\widetilde \Phi_N(E)^\intercal]$
\[ s_W(\pi_{F}^*U_N(E) \pi_{F_+}) \geqslant e^{-\frac{\epsilon}{25} N}~.  \]
\end{rmk}
\begin{proof} We abbreviate $\Sigma = \Sigma_N(E)$, $V = V_N(E)$, and $\gamma_j = \gamma_j(E)$. On the other hand, the constants with $\epsilon$ not explicitly present in the notation will be uniform in $\epsilon \to + 0$.

Clearly, $s_j(\pi_{F_+ }V^\intercal \pi_F^*)  \leqslant \| \pi_{F_+ }V^\intercal \pi_F^* \| \leqslant 1$. Hence it will suffice to show that on $\Omega_\epsilon^F[\widetilde \Phi_N(E)]$
\begin{equation}\label{eq:lwbdprod}
\prod_{k=1}^W s_k(\pi_{F_+ }V^\intercal \pi_F^*) \geqslant e^{-\epsilon N}~.
\end{equation}

Let $\Sigma^+$ be the diagonal matrix obtained by setting the $(k, k)$ matrix entries of $\Sigma$ to zero for $k > W$. Then on $\Omega_\epsilon^F[\widetilde \Phi_N(E)]$ we have 
\[ \| \Sigma -  \Sigma^{+} \|  \leqslant \exp(  - c N)  \]
(with $c>0$ uniform in $\epsilon \to +0$).
Thus $s_j(\widetilde \Phi_N \pi_F^*) = s_j(\Sigma V^\intercal \pi_F^*)$ satisfies
\[ | s_j(\widetilde \Phi_N  \pi_F^*) - s_j(\Sigma^{+} V^\intercal \pi_F^*)| \leqslant e^{-cN}~.\]
Observing that $s_j(\Sigma^{+} V^\intercal \pi_F^*) = s_j(\widehat\Sigma^{+}\pi_{F_+}  V^\intercal \pi_F^*)$, where $\widehat \Sigma^+ = \pi_{F_+} \Sigma^+ \pi_{F_+}^*$, and that 
\[  s_j(\widetilde \Phi_N \pi_F^*) \geqslant e^{(\gamma_j - \frac{\epsilon}{100W})N}\]
on $\Omega_\epsilon^F$~, we get (for sufficiently large $N$):
\[  s_j(\widehat \Sigma^+ \pi_{F_+} V^\intercal \pi_F^*) \geqslant e^{(\gamma_j - \frac{\epsilon}{50W})N}~,  
\quad
\prod_{k=1}^j s_k(\widehat \Sigma^+ \pi_{F_+}  V^\intercal \pi_F^*) \geqslant e^{(\gamma_1 + \cdots + \gamma_j -\frac{\epsilon}{50})N}~.\]
On the other hand, using the submultiplicativity of the operator norm and the equalitiy between the norm of the $j$-th wedge power of a matrix and the product of its $j$ top singular values, we have
\[\begin{split}  \prod_{k=1}^j s_k(\widehat \Sigma^+ \pi_{F_+}  V^\intercal \pi_F^*) &\leqslant   \prod_{k=1}^j s_k (\widehat \Sigma^+)  \times \prod_{k=1}^j s_k(\pi_{F_+}  V^\intercal \pi_F^*)\\
&\leqslant e^{(\gamma_1 + \cdots + \gamma_j + \frac{\epsilon}{100})N}  \prod_{k=1}^j s_k(\pi_{F_+}  V^\intercal \pi_F^*)~, 
\end{split}\]
whence 
\[ \prod_{k=1}^j s_k(\pi_{F_+}  V^\intercal \pi_F^*) \geqslant e^{-\frac{\epsilon}{25}N}~, \quad 1 \leqslant j \leqslant W~, \]
thus concluding the proof of (\ref{eq:lwbdprod}) and of the claim.
\end{proof}

\subsection{Wegner-type estimate: proof of Proposition~\ref{prop:Weg}}\label{s:Weg}

Let us first show that for any $i \in [-N, N]$
\begin{equation}\label{eq:diag-G}
\mathbb P \left\{ \|G_E[H_{-N,N]}](i,i) \| \geqslant e^{\epsilon N} \right\} \leqslant C_\epsilon e^{-c_\epsilon N }~. \end{equation}
By Claim~\ref{cl:formG}, 
\[ G_E[H_{[-N, N]}](i, i) = \left( \Psi_{i+1}^+ (\Psi_i^+)^{-1} - \Psi_{i+1}^- (\Psi_i^-)^{-1}\right)^{-1} L_i^{-1}~,\]
where
\[\begin{split} \binom{\Psi_{i+1}^+}{\Psi_i^+} &= \Phi_{i+1,N+1} \binom{0}{\mathbbm 1} = \left( \begin{array}{cc} \mathbbm 1 & 0 \\ 0 & L_i^{- \intercal} \end{array} \right)  \widetilde\Phi_{i+1,N+1}  \binom{0}{ L_N^\intercal}~, 
\\
\binom{\Psi_{i+1}^-}{\Psi_i^-} &= \Phi_{i+1,-N} \binom{\mathbbm 1}{0} = \left( \begin{array}{cc} \mathbbm 1 & 0 \\ 0 & L_i^{- \intercal} \end{array} \right)  \widetilde\Phi_{i+1,-N}  \binom{\mathbbm 1}0~.
\end{split}\]
Hence 
\[  G_E[H_{[-N, N]}](i, i)  = L_i^{-\intercal} \left( X^+ - X^- \right)^{-1} L_i^{-1}~, \]
where
\[ X^+ = (\widetilde\Phi_{i+1,N+1} )_{12}   (\widetilde\Phi_{i+1,N+1})_{22}^{-1}~,  \quad
X^- =  (\widetilde \Phi_{i+1,-N} )_{11}  (\widetilde \Phi_{i+1,-N})_{21}^{-1}~,\] 
and the subscripts $11$ and $21$ represent extracting the corresponding $W\times W$ blocks from a $2W \times 2W$ matrix (i.e.\ $Y_{11} = \pi_{F_+} Y \pi_{F_+}^*$, $Y_{21} = \pi_{F_-} Y \pi_{F_+}^*$, in the notation of (\ref{eq:Fpm})).
Both matrices $X^\pm$ are Hermitian,  as  follows from the symplectic property of the transfer matrices.

Without loss of generality we can assume that $i \geqslant 0$. We shall prove that 
\[ \mathbb P \left\{ s_W(X^+ - X^-) \leqslant e^{-\epsilon N} \,| \, X^+ \right\} \leqslant C_\epsilon e^{-c_\epsilon N }~. \]
To this end, denote
\[ F = \left\{ \binom{x}{y} \in \mathbb R^{2W} \, : \, y = - X^+ x \right\}~. \]
In the notation of Claim~\ref{cl:unit}, consider the transfer matrix $\widetilde\Phi_{i+1,-N}$, and let  
\[ \Omega_\epsilon = \Omega_\epsilon^F[\widetilde \Phi^*] \cap \Omega_\epsilon^{F_+} [\widetilde \Phi^*]  \cap \Omega_\epsilon^{F_-} [\widetilde \Phi^*] \cap \Omega_\epsilon^{F_+} [\widetilde \Phi ]   \]
 (note that $\widetilde \Phi_{i+1,-N}$ is independent of $X^+$ and thus also of $F$).   It suffices to show that on $\Omega_\epsilon$
\begin{equation}\label{eq:sWdiffX} s_W(X^+ - X^-) \geqslant e^{-\frac\epsilon2 N}~. \end{equation}
Let us write the singular value decomposition of $\widetilde \Phi = \widetilde\Phi_{i+1,-N}$ in block form:
\[ 
\left(\begin{array}{cc}
\widetilde \Phi_{11} & \widetilde \Phi_{12} \\ \widetilde \Phi_{21} & \widetilde \Phi_{22} \end{array}\right)
=
\left(\begin{array}{cc}
U_{11} &  U_{12} \\  U_{21} &  U_{22} \end{array}\right)
\left(\begin{array}{cc}
 \widehat \Sigma^+ &   \\   &  \widehat \Sigma^- \end{array}\right)
\left(\begin{array}{cc}
 V_{11}^\intercal &  V_{21}^\intercal \\  V_{12}^\intercal &  V_{22}^\intercal \end{array}\right)
\]
whence on $\Omega_\epsilon$
\[ \| \widetilde \Phi_{11} - U_{11} \widehat \Sigma^+ V_{11}^\intercal \|~, \| \widetilde \Phi_{21} - U_{21} \widehat \Sigma^+ V_{11}^\intercal \| \leqslant e^{-cN}~. \]
Further, by Claim~\ref{cl:unit} we have on $\Omega_\epsilon$:
\begin{equation}\label{eq:sWblocks} s_W(U_{11})~, s_W(U_{21})~, s_W(V_{22}) \geqslant e^{-\frac{\epsilon}{25} N}~. \end{equation}
Let us show that
\begin{equation}\label{eq:resolve}
\| X^- - U_{11} U_{21}^{-1} \| \leqslant e^{-c'N}~.
\end{equation}
To this end, start with the relation
\[ X^- = (U_{11} \widehat \Sigma^+ V_{11}^\intercal + E_1) (U_{21} \widehat \Sigma^+ V_{11}^\intercal+ E_2)^{-1}, \quad \| E_1\|, \|E_2\| \leqslant e^{-cN}~.\]
In view of the bound 
\[ s_W(U_{21} \widehat \Sigma^+ V_{11}^\intercal) \geqslant s_W(U_{21}) s_W(\widehat \Sigma^+) s_W(V_{11}^\intercal) \geqslant e^{+c_1 N}~,\]
we can set $E_2' = E_2 (U_{21} \widehat \Sigma^+ V_{11}^\intercal)^{-1}$ and rewrite
\[ (U_{21} \widehat \Sigma^+ V_{11}^\intercal+ E_2)^{-1} = (U_{21} \widehat \Sigma^+ V_{11}^\intercal)^{-1} (1 + E_2')~, \quad  \|E_2'\| \leqslant e^{-c_2N}~,\] 
which implies (\ref{eq:resolve}).

Now, the matrix $X^+$ is symmetric, therefore $x - X^+ y = 0$ for $\binom{x}{y} \in F^\perp$, whence for any  $\binom{x}{y} \in \mathbb R^{2W}$
\[ x - X^+ y =  (\mathbbm 1 \,\,\mid \,\, - X^+) \pi_F^* \pi_F  \binom{x}{y}\]
(where the first term is a $1 \times 2$ block matrix).
Therefore we have, by 
another application of Claim~\ref{cl:unit}:
\[\begin{split} s_W(U_{11} - X^+ U_{21})  &= s_W((\mathbbm 1 \,\,\mid \,\, - X^+) \pi_F^* \pi_F U \pi_{F_+}^*)  \\ &\geqslant s_W((\mathbbm 1 \,\,\mid \,\, - X^+) \pi_F^*) s_W(\pi_F U \pi_{F_+}^*) \geqslant  s_W(\pi_F U \pi_{F_+}^*) \geqslant e^{-\frac{\epsilon}{25}N}~. \end{split} \]
This, together with (\ref{eq:resolve}) and (\ref{eq:sWblocks}), concludes the proof of (\ref{eq:sWdiffX}), and of (\ref{eq:diag-G}). 

Now we consider the elements $G_E[H_{[-N,N]}](i, i\pm 1)$. We have: 
\[G_E[H_{[-N,N]}](i, i\pm 1) =\Psi_{i+1}^\pm (\Psi_i^\pm)^{-1}  G_E[H_{[-N,N]}](i, i)~.\]
The norm of $ G_E[H_{[-N,N]}](i, i)$ is controlled by (\ref{eq:diag-G}), whereas $\Psi_{i+1}^\pm (\Psi_i^\pm)^{-1} = L_i^{-1} X^\pm$ are controled using (\ref{eq:resolve}) and Claim~\ref{cl:unit}.
\qed

\section{On generalisations}

\paragraph{Other distributions} The assumptions (A)--(C) in Theorems~\ref{thm:1} and \ref{thm:2} can probably be relaxed. Instead of a finite fractional moment in (A), it should be sufficient to assume the existence of a sufficiently high logarithmic moment:
\[ \mathbb{E} (\log_+^A  \|V_0 \|  + \log_+^A  \| L_0\|  +  \log_+^A \|L_0^{-1}\| ) < \infty\] 
for a sufficiently large $A > 1$. To carry out the proof under this assumption in place of (A), one would need appropriate versions of large deviation estimates 
for random matrix products. 

As we saw in the previous section, the  r\^ole of the assumptions (B)--(C) is to ensure that the conditions of the Goldsheid--Margulis theorem \cite{GM} are satisfied. That is, our argument yields the following:
\begin{thm}\label{thm:3} Let $I \subset \mathbb R$ be a compact interval. Assume (A) and that for any $E \in I$ the group generated by 
\[ \left\{ Q(L, E\mathbbm 1 - V)\right\}_{L \in \mathcal L, \, V \in \mathcal V}\]
is Zariski-dense in $\Sp(2W, \mathbb R)$. Then:
\begin{enumerate}
\item The spectrum of $H$ in $I$ is almost surely pure point, and 
\begin{equation}\label{eq:thm'} \mathbb P \left\{ \forall (E, \psi) \in \mathcal E[H] \,\,\, E \in I \Longrightarrow \limsup_{x \to \pm\infty} \frac{1}{|x|} \log \|\psi(x)\| \leqslant - \gamma_W(E)\right\}   =1~;\end{equation}
\item for any compact subinterval $I' \subset I$ (possibly equal to $I$) one has:
\begin{equation}\label{eq:thm2-bis}\mathbb P\left\{  \limsup_{x \to \pm \infty} \frac{1}{|x|} \log Q_I(x, y) \leqslant - \inf_{E \in I} \gamma_W(E) \right\} = 1~.\end{equation}
\end{enumerate}
\end{thm}
As we saw in the previous section, the second condition of this theorem is implied by our assumptions (B)--(C). 
Most probably, weaker assumptions should suffice, and, in fact, we believe that the conclusions of Theorems~\ref{thm:1} and \ref{thm:2} hold as stated without the assumption (B). A proof would require an appropriate generalisation of the results of Goldsheid \cite{G95}. 

Another interesting class of models appears when $V_x \equiv 0$. The complex counterpart of this class, along with a generalisation in which the distribution of $L_x$ depends on the parity of $x$, has recently been considered by Shapiro \cite{Sh}, in view of applications to topological insulators. An interesting feature of such models is that the slowest Lyapunov exponent $\gamma_W(E)$ may vanish at $E=0$. This circle of questions (in partiular, the positivity of the smallest Lyapunov exponent and Anderson localisation) is studied in \cite{Sh} under the assumption that the distribution of $L_0$ in $\GL(W, \mathbb C)$ is regular. In order to extend the results of \cite{Sh} (for matrices  complex entries) to singular distributions, one would first need an extension of \cite{GM} to the Hermitian symplectic group. 

Returning to the (real) setting of the current paper, assume that (B)--(C) are replaced with
\begin{enumerate}
\item[(B$'$)] the group generated by $\mathcal L$ is Zariski-dense in $\GL(W, \mathbb R)$;
\item[(C$'$)] $V_x \equiv 0$.
\end{enumerate}
Along the arguments of \cite{Sh}, one can check that the conditions of \cite{GM} hold for any $E \neq 0$. From Theorem~\ref{thm:3}, one deduces that the conclusion of Theorem~\ref{thm:1} holds under the assumptions (A), (B$'$), (C$'$), whereas the conclusion (\ref{eq:thm2-bis}) of Theorem~\ref{thm:2} holds for compact intervals $I$ not containing $0$. If $\gamma_W(0)= 0$, (\ref{eq:thm2-bis}) is vacuous for $I \ni 0$. If $\gamma_W(0) > 0$, (\ref{eq:thm2-bis}) is meaningful and probably true for such intervals, however, additional arguments are required to establish the  large deviation estimates required for the proof.

Finally, we note that Theorem~\ref{thm:3} remains valid if the independence assumption is relaxed as follows: $\{(V_x, L_x)\}_{x \in \mathbb Z}$ are jointly independent (i.e. we can allow dependence between $V_x$ and the corresponding $L_x$).

\paragraph{The half-line} Similar results can be established for random operators on the half-line. For simplicity, we focus on the case $L_x \equiv \mathbbm 1$. Fix   a Lagrangian subspace $F \subset \mathbb R^{2W}$, and consider the space $\mathcal H_F$ of square-summable sequences $\psi: \mathbb Z_+ \to \mathbb C^W$ such that $\binom{\psi(1)}{\psi(0)} \in F$. Define an operator $H_F$ acting on $\mathcal H_F$ so that
\[ (H_F\psi)(x) = L_x \psi(x+1) + V_x \psi(x) + L_{x-1}^\intercal \psi(x-1)~, \quad x \geq 1 \]
(see e.g.\ \cite{Atk} for details). 
\begin{thm}\label{thm:4}
Fix  a Lagrangian subspace $F \subset \mathbb R^{2W}$. Under the assumptions (A) and (C) with $L_x \equiv 1$,  the spectrum of $H_F$ in any compact interval $I$ is almost surely pure point, and 
\begin{equation}\label{eq:thm''} \mathbb P \left\{ \forall (E, \psi) \in \mathcal E[H_F] \,\,\, E \in I \Longrightarrow \limsup_{x \to  \infty} \frac{1}{|x|} \log \|\psi(x)\| \leqslant - \gamma_W(E)\right\}   =1~.\end{equation}
\end{thm}
\begin{rmk} \hfill
\begin{enumerate}
\item For general $L_x$, the boundary condition has  to be prescribed in a different way. However, in the Dirichlet case $F = F_+$ the result holds as stated for general $L_x$ satisfying (A)--(B).
\item Combining the proof of Theorem~\ref{thm:2} with the additional argument described below, one can also prove dynamical localisation.
\item For $W=1$, a result of Kotani \cite{Kot} implies that the operator  $H_F$ has pure point spectrum for almost every boundary condition $F$; a similar statement is valid for $W>1$. As to fixed (deterministic) boundary conditions, the only published reference known to us is the work of Gorodetski--Kleptsyn \cite{GorKl}, treating Schr\"odinger operators in $W=1$ with Dirichlet boundary conditions.
\item The event of full probability provided by Theorem~\ref{thm:4} depends on the boundary condition $F$. And indeed, a result of Gordon \cite{Gor} implies that (almost surely) there exists a residual set of initial conditions $F$ for which the spectrum of $H_F$ is not pure point (and in fact has only isolated eigenvalues).
\end{enumerate}
\end{rmk}

\begin{proof}[Sketch of proof of Theorem~\ref{thm:4}] We indicate the necessary modifications with respect to the proof of Theorem~\ref{thm:1}.   First, we modify the definition (\ref{eq:nonres}) of $\Res(\tau, E, N)$ as follows: $x \geqslant N+1$ is said to be $(\tau,E,N)$-non-resonant ($x \notin\Res(\tau, E,N)$)  under the same condition 
\begin{equation}\label{eq:nonres'} 
\|G_E[H_{[x-N, x+N]}](x, x\pm N)\| \leqslant e^{-(\gamma_W(E) - \tau)N}~,\end{equation}
 while $x \in \{1, \cdots, 2 N\}$ is said to be $(\tau,E,N)$-non-resonant if 
\begin{equation}\label{eq:nonres''}\det (\pi_F \Phi_N(E)^* \Phi_N(E) \pi_F^*) \geqslant e^{2 (\gamma_1(E) + \cdots + \gamma_{W}(E) - \tau)N}\end{equation}
(this condition does not depend on $x$, and only depends on the restriction of the operator to $[1, N]$).
We claim that Proposition~\ref{prop:main} is still valid: 
\begin{equation}\label{eq:prop-main-bc} \mathbb P \left\{ \max_{E \in I} \diam (\Res(\tau, E, N) \cap [1, N^2]) > 2N \right\} \leqslant C e^{-cN}~. \end{equation}
To prove this estimate, it suffices to show that for any $1 \leqslant x < y \leqslant N^2$ with $|y-x|>2N$ one has
\begin{equation}\label{eq:prop-main-bc;} \mathbb P \left\{  \exists {E \in I} : \,   x, y \in \Res(\tau, E, N)   \right\} \leqslant C e^{-cN}~. \end{equation}
The case $x, y > 2N$ is covered by the current Proposition~\ref{prop:main}. If $x \leqslant 2N$ and $y > 2N$ the events $x \in \Res(\tau, E, N)$ and $y \in \Res(\tau, E, N)$ are independent; the probability that $x \in \Res(\tau, E, N)$ is exponentially small due to the large deviation estimate  (\ref{eq:LDP}), and the collection of $E$ violating (\ref{eq:nonres''}) is the  union of $\leqslant N^{2W}$ intervals. From this point the proof of (\ref{eq:prop-main-bc;}) mimics the argument in the proof of  Proposition~\ref{prop:main}.

As in the proof of Theorem~\ref{thm:1}, let $\psi$ be a generalised solution at energy $E \in I$ given by Schnol's lemma, $x^{-1} \log \| \psi(x) \| \to 0$. Letting $u_x = \binom{\psi(x)}{\psi(x-1)}$, we have 
\[ \| \Phi_N (E) u_1 \| \leqslant  e^{\tau N }~,   \]
hence for sufficiently large $N$ one has 
\[ s_W(\Phi_N(E) \pi_F^*) \leqslant   e^{2 \tau N}~.\]
On the other hand, on an event of full probability one has for all $E \in I$ and all sufficiently large $N$
\[ (s_1 \cdots s_{W-1})(\Phi_N(E) \pi_F^*) \leqslant (s_1 \cdots s_{W-1})(\Phi_N(E)) \leqslant e^{(\gamma_1(E) + \cdots + \gamma_{W-1}(E) + \tau) N }   \]
due to a version of the Craig--Simon theorem \cite{CS} (cf.\ \cite[Lemma 2.2]{GS}). This implies 
\[ (s_1 \cdots s_{W})(\Phi_N(E) \pi_F^*) \geqslant e^{(\gamma_1(E) + \cdots + \gamma_{W-1}(E)  + 3\tau) N }~,   \]
which contradicts (\ref{eq:nonres''}) when $\tau > 0$ is  small enough. Thus for $N$ large enough 
\[ \Res(\tau, E,N) \cap [N+1, N^2] = \varnothing~, \]
and thus $\psi$ decays exponentially as in the proof of Theorem~\ref{thm:1}.
\end{proof}

\end{document}